\documentclass[letterpaper,11pt,onecolumn]{article}
\usepackage{graphicx}
\usepackage{psfrag}
\usepackage{cite}
\usepackage{amssymb}
\usepackage{amsmath}
\usepackage{algorithm}
\usepackage{algpseudocode}
\usepackage{ulem}	
\usepackage[latin1]{inputenc} 
\usepackage{dsfont}

\newenvironment{proof}[1]{\textbf{Proof. }#1}{\hfill $\Box$}

\newtheorem{theorem}{Theorem}
\newtheorem{lemma}{Lemma} 

\newtheorem{proposition}{Proposition}

\newtheorem{remark}{Remark}
\newtheorem{example}{Example}

\newcommand{\be}{\begin{equation}}
\newcommand{\ee}{\end{equation}}
\newcommand{\ba}{\begin{array}}
\newcommand{\ea}{\end{array}}
\newcommand{\rr}{\mathbb{R}}
\newcommand{\nn}{\mathbb{N}}

\def\cC{{\mathcal C}}
\def\cM{{\mathcal M}}
\def\cB{{\mathcal B}}
\def\cS{{\mathcal S}}
\def\cZ{{\mathcal Z}}

\def\cP{{\mathcal P}}
\def\cV{{\mathcal V}}
\def\cW{{\mathcal W}}

\def\cAS{{\Pi}}

\def\rad{\mbox{rad}}
\newcommand{\ds}{\displaystyle}

\newcommand{\no}[1]{}

\usepackage{color}  

\title{Adaptive Threshold Selection for Set Membership State Estimation with Quantized Measurements} 

\date{\today}
\author{Marco Casini, Andrea Garulli, Antonio Vicino}

\begin{document}

\normalem  

\maketitle

\begin{abstract}
State estimation for discrete-time linear systems with quantized measurements is addressed. By exploiting the set-theoretic nature of the information provided by the quantizer, the problem is cast in the set membership estimation setting. Motivated by the possibility of suitably tuning the quantizer thresholds in sensor networks, the optimal design of adaptive quantizers is formulated in terms of the minimization of the radius of information associated to the state estimation problem. The optimal solution is derived for first-order systems and the result is exploited to design adaptive quantizers for generic systems, minimizing the size of the feasible output signal set. Then, the minimum number of sensor thresholds for which the adaptive quantizers guarantee asymptotic boundedness of the state estimation uncertainty is established. Threshold adaptation mechanisms based on several types of outer approximations of the feasible state set are also proposed. The effectiveness of the designed adaptive quantizers is demonstrated on numerical tests involving a specific case study and randomly generated systems, highlighting the trade off between the resulting estimation uncertainty and the computational burden required by recursive set approximations.
\end{abstract}

\section{Introduction}

The effect of quantization in control systems has been investigated since long time (see, e.g., \cite{Cur70,nfze07} and references therein). In particular, it is well known that adaptive quantizers can be employed to improve the performance of the control loop \cite{bl00,em01,ps01,as08}. 
Due to the steadily increasing diffusion of networked systems with limited communication rate, the role of quantization has been intensively studied also in system identification \cite{wyzz10,wkc09,ggamw11,cgv12b,cpr13,bhp17} and state estimation \cite{wb97,sw00,fds09,zwm16}.

In networked systems, such as wireless sensor networks or industrial control networks, quantization thresholds can be adapted to reduce estimation errors and improve the system performance.
In fact, while the thresholds of physical sensors are often difficult to trigger, quantization for communication allows the user 
to choose how to partition the signal range in subsets, which corresponds to suitably select the quantizer thresholds \cite[Chapter 1]{wyzz10}. This has motivated a growing interest in the design of adaptive quantizers for system identification \cite{agy07,wyzz08,tsu09,myf13,you15} and distributed sensing \cite{pwo01,fl08}.
In \cite{agy07}, adaptation is performed by using a noise-shaping coder, a solution widely employed in delta-sigma data converters \cite{dsp15}. Time-varying thresholds are obtained in \cite{pwo01} by adding a control input before sensor quantization. Three different threshold adaptation schemes for distributed parameter estimation are proposed in \cite{fl08}. They require the solution of an online optimization problem, while simpler recursive adaptation schemes are proposed in \cite{myf13,you15}.  

In this paper, the problem of designing adaptive quantizers will be addressed in the context of set membership state estimation for discrete-time linear systems. The aim is to find an adaptation scheme for the sensor thresholds which minimizes the set-theoretic uncertainty associated to the state estimates.  

\subsection{Related work}

Set membership state estimation has a long history, whose origin dates back to more that 50 years ago \cite{S68,BR71}. The main idea is as follows: when the system is affected by unknown-but-bounded (UBB) process disturbance and measurement noise, it is possible to define the so-called \emph{feasible state set}, i.e., the set of all the states that are compatible with the available input-output measurements and the disturbance bounds. An advantage of this approach is that all the elements of the feasible set can be equivalently considered as pointwise estimates of the state, depending on the specific application at hand. Different measures of the size of the feasible set can be used to quantify the uncertainty affecting the estimate. On the other hand, a major problem in set membership estimation is that feasible sets rapidly become too complex to be computed exactly. Even in the case of linear dynamic systems, feasible sets are usually polytopes with an increasing number of facets and vertices. This has generated a rich literature in which researchers have proposed a number of techniques for approximating the feasible sets. The common idea underlying these approaches is to choose a family of sets of limited complexity and compute a set within the family that contains the true feasible set. Popular approximating sets incluse ellipsoids \cite{gnkh96,dwp01,ec01}, parallelotopes \cite{CGZ96}, zonotopes \cite{abc05,Com15, wwpc19,ar21} and constrained zonotopes \cite{srmb16,rk22,rlrr22}. The interested reader is referred to \cite{drt22} for a detailed survey on the subject.

The effect of measurement quantization in state estimation has been mostly addressed in the stochastic framework. Since the seminal work in \cite{wb97}, the objective has been to design the quantizer parameters so that the estimation error possesses some desirable properties (typically, a bounded covariance). 
By leveraging the set-theoretic nature of the information provided by quantized sensor measurements, a natural alternative is to cast the state estimation problem in a deterministic setting. 
Recently, this approach has been proposed for state estimation based on binary sensor measurements, in the presence of UBB disturbances, leading to deterministic bounds on the estimation errors \cite{bcg17,zcy20}. 
A further step is to formulate the state estimation problem in the set membership estimation framework, which allows to characterize the uncertainty associated to the state estimates in terms of feasible state sets. Quite interestingly, the connection with the set membership approach was already present in \cite{sw00}, which however considered a very specific system class (cascaded integrators). 
A set-valued observer for linear systems with quantized measurements was first proposed in \cite{hgw04}.
In \cite{zol18} a set membership state estimation algorithm is presented, assuming that quantized measurements of all the state variables are available and that such measurements always remain in the quantizer range. In all these works, the quantizer thresholds are fixed a priori. In \cite{cgv23}, a strategy for adapting the threshold of a binary sensor to minimize the size of the feasible state set is proposed.

\subsection{Paper contribution}

In this paper, the state estimation problem with quantized measurements is addressed for linear discrete-time systems affected by UBB disturbances. 
The problem is cast in the set membership estimation framework. Assuming that the thresholds of the quantizer can be adapted online, the optimal design of the quantizer parameters (range and resolution) is addressed, in order to minimize the \emph{radius of information}, i.e., the radius of the feasible state set corresponding to the worst-case realization of the quantized measurements. To the best of our knowledge, this problem has never been studied in a set-theoretic framework.
First, a complete solution is provided for first-order systems. Then, such a result is exploited to design adaptive quantizers for general $n$-th order systems, by minimizing the radius of the feasible output signal set, namely, the set of feasible values of the (unaccessible) output signal. Moreover, it is shown that such a design guarantees asymptotic boundedness of the feasible state set, and hence of the uncertainty affecting the state estimate, provided that the quantizer has a sufficient number of thresholds. 
The threshold adaptation mechanism is also applied to outer approximations of the feasible state sets, provided by recursive approximation algorithms based on parallelotopes, zonotopes and constrained zonotopes. In particular, analytical expressions of the time-varying thresholds are derived for parallelotopes and zonotopes, thus making the computational burden of online adaptation negligible. The effectiveness of the proposed techniques for designing and adapting the quantizer parameters is demonstrated via numerical simulations on a specific case study and on randomly generated systems with multiple process disturbances and quantized measurements.

Summing up, the contributions of the paper are:
\begin{itemize}
\item[i)] the optimal design of adaptive quantizers for set membership state estimation of linear systems with quantized  measurements;
\item[ii)] a result showing that optimal adaptive quantizers guarantee asymptotic boundedness of the feasible state set, even in the case of unstable systems; 
\item[iii)] an effective and computationally viable way to select the quantizer thresholds, based on commonly employed outer approximations of the feasible state set;
\item[iv)] empirical evidence, based on extensive numerical simulations, that the proposed approach outperforms the results obtained by set membership estimators based on fixed thresholds.
\end{itemize}

The rest of the paper is organized as follows. Section \ref{sec:preliminaries} introduces some preliminary material on set manipulations. The optimal threshold selection problem is formulated in Section \ref{sec:problem_formulation}. The problem solution is derived in Section \ref{sec:1st_order} for first-order systems, and then extended to linear systems of arbitrary dimension in Section \ref{sec:nth_order}.
A threshold selection mechanism based on outer approximations of the feasible state set is illustrated in Section \ref{sec:approx}. Numerical simulations are presented in Section \ref{sec:numex}, while Section \ref{sec:conc} contains some concluding remarks and suggests possible future developments.

\section{Definitions and preliminaries}\label{sec:preliminaries}

$D=\mbox{diag}(v)$ is a diagonal matrix with diagonal equal to $v$. The column vector of ones is $\mathds{1}=[1~ 1 \dots 1]'$.
\\
Given sets $\cV,\cW$ and a matrix $A$, we define the set operations
\begin{eqnarray}
\cV+\cW&=&\{z:~ z=v+w,~ v \in \cV, ~w \in \cW\},\\
A\cV&=&\{z:~ z=Av,~ v \in \cV\}.
\end{eqnarray}
We denote by $\|v\|_p$ the $p$-norm of $v$ (we omit the subscript when $p=2$). $\cB_\infty$ is the unit ball in the infinity norm, i.e. $\cB_\infty=\{x:~\|x\|_\infty \leq 1\}$.
\\
The radius of a set $\cV$ is defined as
\begin{equation}
\mbox{rad}(\cV)= \inf_z \sup_{v\in\cV} \| z-v \|.
\label{eq:rad}
\end{equation}
The argument $\hat{z}$ of the infimum in \eqref{eq:rad} is the center of $\cV$.
\\
Given $p\in\rr^n$ and $c \in\rr$, a \emph{strip} is defined as 
\begin{equation}
\label{eq:strip}
\cS(p,c)=\{x\in\rr^n:~|p'x-c| \leq 1\}.
\end{equation}
A non-degenerate \emph{parallelotope} is the intersection of $n$ linearly independent strips. Given $c\in\rr^n$ and a nonsingular matrix $P\in\rr^{n \times n}$, it is defined as
\begin{equation}
\cP(P,c)=\{x\in\rr^n:~\|Px-c\|_\infty\leq 1\}=\{x\in\rr^n:~x=x_c+T\alpha, ~\|\alpha\|_\infty\leq 1\}
\label{eq:parallelotope}
\end{equation}
where $T=P^{-1}$ and $x_c=P^{-1}c$.
Its radius is given by
\begin{equation}
\rad({\cP(P,c)}) = \max_{w \in \cB_\infty} \| P^{-1} w \|.
\label{eq:radpar}
\end{equation}
Moreover, for any vector $v\in\rr^n$, it holds
\begin{equation}
\sup_{x \in\cP} v'x - \inf_{x \in\cP} v'x = 2 \|v'P^{-1}\|_1.
\label{eq:projpar}
\end{equation}
A \emph{zonotope} is an extension of a parallelotope defined as
\begin{equation}
\cZ(T,x_c)=\{x\in\rr^n:~x=x_c+T\alpha, ~\|\alpha\|_\infty\leq 1\}
\label{eq:zonotope}
\end{equation}
where $T\in\rr^{n\times m}$ and $\alpha \in\rr^m$. The columns of $T$ are the generators of the zonotope and its order is $o=\frac{m}{n}$.
\\
A \emph{constrained zonotope} is an extension of a zonotope defined as
\begin{equation}
\cC \cZ(T,x_c,A,b)=\{x\in\rr^n:~x=x_c+T\alpha, ~\|\alpha\|_\infty\leq 1,~A\alpha=b\}
\label{eq:constrzonotope}
\end{equation}
where $A\in\rr^{n_c\times m}$ and $b \in\rr^{n_c}$. The order of a constrained zonotope is defined as $o=\frac{m-n_c}{n}$, where $n_c$ is the number of equality constraints in \eqref{eq:constrzonotope}.
\begin{proposition}
\label{prop1}
Given a strip $\cS(p,c) \subset \rr^n$, a nonsingular matrix $A \in\rr^{n\times n}$, a matrix $G\in\rr^{n \times m}$ and a constant $\gamma>0$
\begin{equation}
A\cS(p,c)+\gamma G\cB_\infty = \{x\in\rr^n: ~|p'A^{-1} x -c| \leq 1+\gamma \|p'A^{-1}G\|_1\}.  
\label{eq:prop}
\end{equation}
\end{proposition}
\begin{proof}
See Appendix \ref{app:prop1}.
\end{proof}

\section{Problem formulation}\label{sec:problem_formulation}

Consider the discrete-time linear system
\begin{eqnarray}
x(k+1) &=& Ax(k) +Gw(k) \label{eq:sys1}\\
z(k)&=&Cx(k)+v(k) \label{eq:sys2}
\end{eqnarray}
where $x(k)\in\rr^n$ is the state vector, $z(k)\in\rr$ is the (inaccessible) output signal, $w(k)\in\rr^m$ is the process disturbance and $v(k)\in\rr$ is the output noise. A single output is assumed only to simplify the presentation; the results can be easily extended to multi-output systems.
The process disturbance $w(k)$ and the output noise $v(k)$ are unknown-but-bounded (UBB) sequences, i.e. they satisfy
\begin{eqnarray}
&& \|w(k)\|_\infty  \leq  \delta_w \vspace*{1mm} \label{eq:ubbw}\\
&& |v(k)|  \leq  \delta_v
\label{eq:ubbv}
\end{eqnarray}
for all $k \geq 0$.

Observations of the output signal are obtained from a quantized sensor with $d$ thresholds $\tau_1,\ldots,\tau_d$, such that 
\begin{equation}
\label{eq:quant}%
y(k)=\sigma(z(k))\triangleq\left\{%
\begin{array}{ll}
0 & \mbox{ ~if $z(k)\le \tau_1$}\\
1 & \mbox{ ~if $\tau_1<z(k)\le \tau_2$}\\
& ~~\vdots  \\
d & \mbox{ ~if $z(k)>\tau_d$}
\end{array}
\right.%
\end{equation}
The quantization is assumed to be uniform in the range $[\tau_1,\tau_d]$, i.e., the sensor  thresholds satisfy
$\tau_{i+1}-\tau_{i}=\Delta$, $i=1,\ldots,d-1$, for some $\Delta>0$. For simplicity of exposition, we assume $d$ is odd (hence, the number of possible outcomes of $y(k)$ is even). By setting $c=(d+1)/2$, the quantizer is defined by choosing the central threshold $\tau_c$ and the resolution $\Delta$. Hence, the quantizer is fully characterized by the parameter vector $\theta=[\tau_c~\Delta]'$.
Then, the thresholds can be expressed in terms of $\theta$ as
\begin{equation}
\tau_i=\tau_c-\left(\frac{d+1}{2}-i\right)\Delta~,~~~i=1,\dots,d.
\label{eq:thresh1}
\end{equation}

According to the set membership estimation paradigm, the information provided by the quantized measurement $y(k)$ is captured by the \emph{feasible measurement set}, i.e., the set of states $x(k)$ which are compatible with $y(k)$.
By using \eqref{eq:sys2} and \eqref{eq:ubbv}, this is given by
\begin{equation}
\begin{array}{rlll}
\cM(k) =&
 \{ x\in\rr^n~: & Cx  \le \tau_1 + \delta_v &~\mbox{ if } y(k)=0;\\
&&  \tau_1-\delta_v< Cx  \le \tau_2 + \delta_v &~\mbox{ if } y(k)=1;\\
&&\vdots &\\[1mm]
&&  Cx> \tau_d-\delta_v   &~\mbox{ if } y(k)=d \}.
\end{array}
\label{eq:cMk}
\end{equation}
The set $\cM(k)$ is a strip in the state space when $1 \leq y(k) \leq d-1$, and a half-space when $y(k)=0$ or $y(k)=d$ (notice that this occurs even in the noiseless case $\delta_v=0$). This justifies the set membership approach to the state estimation problem.

The \emph{feasible state set} $\Xi(k|k)$, i.e. the set of all the state vectors $x(k)$ which are compatible with the quantized measurements collected up to time $k$, is defined by the recursion 
\begin{eqnarray}
\Xi(k|k) &=& \Xi(k|k-1) \bigcap \cM(k) \label{eq:fss1}\\
\Xi(k+1|k) &=& A \Xi(k|k) +\delta_wG\cB_\infty \label{eq:fss2}
\end{eqnarray}
which is initialized by the set $\Xi(0|-1)$, containing all the feasible initial states $x(0)$. Hereafter, we will denote $\rho(k)=\rad(\Xi(k|k))$. 
By choosing the center of $\Xi(k|k)$ as a pointwise estimate of the state $x(k)$, $\rho(k)$ represents the maximum uncertainty associated to the estimate.

If the parameter vector $\theta$ of the quantizer is constant, the state estimation problem may be in general ill-posed. In fact, it may happen that $\cM(k) \supseteq \Xi(0|-1)$ for all $k \geq 0$, thus meaning that the initial uncertainty is never reduced, or that the feasible set tends to become asymptotically unbounded, as in the following example.
\begin{example}
Let $n=1$, $A=G=C=1$ , $\delta_w=1$, $\delta_v=0$. Assume $x(0)=\frac{1}{2}$ and $w(k)=1$, $\forall k \geq 0$, so that $z(k)=x(k)=\frac{1}{2}+k$. Consider the quantizer in \eqref{eq:quant} with $d=3$, $\tau_1=-1$, $\tau_2=0$, $\tau_3=1$. Then, $y(k)=3$, $\forall k \geq 1$, and hence $\cM(k)=\{x:~x>1\}$. By choosing $\Xi(0|-1)=[0,1]$ and applying the recursion \eqref{eq:fss1}-\eqref{eq:fss2}, one gets the sequence of feasible state sets $\Xi(k|k)=[1,k]$ which is asymptotically unbounded.
\end{example}

For the reasons exposed above, it is of interest to adapt the quantizer parameters in order to minimize the size of the feasible state set. 
Let us assume that at each time $k$ it is possible to select the parameters of the quantizer $\theta(k)=[\tau_c(k)~\Delta(k)]'$. The aim is to choose $\theta(k)$ online in order to minimize the size of the resulting feasible set $\Xi(k|k)$. Since $\cM(k)$ depends on the actual value taken by the output $y(k)$, which in turn depends on the feasible predicted states $\Xi(k|k-1)$ and the realization of the measurement noise $v(k)$, a meaningful objective is to minimize the radius of $\Xi(k|k)$ with respect to the worst-case value taken by $y(k)$, i.e., to select
\begin{equation}
\theta^*(k) = \arg \inf_{\theta} \sup_{y(k)} \rho(k;\theta,y(k))
\label{eq:thsel}
\end{equation}
where $\rho(k;\theta,y(k)) = \rad \left(\Xi(k|k-1) \bigcap \cM(k) \right)$, in which the dependence of the radius on both $\theta$ and $y(k)$, through $\cM(k)$ in \eqref{eq:cMk}, have been made explicit.  
Notice that in the information-based complexity framework, this corresponds to trigger the quantizer in order to minimize the so-called \emph{radius of information} \cite{TWW88,MV91a}. In the next section, the solution of problem \eqref{eq:thsel} is derived for first-order systems.

\section{First-order systems}
\label{sec:1st_order}

Let $n=m=1$. Without loss of generality, consider the system
\begin{eqnarray}
x(k+1) &=& ax(k) +w(k) \label{eq:sys1d1}\\
z(k)&=& x(k)+v(k) \label{eq:sys1d2}
\end{eqnarray}
with $|w(k)| \leq \delta_w$, $|v(k)| \leq \delta_v$, $\forall k \geq 0$.
Moreover, let us assume $a>0$. If $\Xi(0|-1)\subset \rr$ is an interval, then all feasible sets $\Xi(k|k)$ and $\Xi(k+1|k)$ are also intervals. Hence, let us adopt the notation $\Xi(k|k)=[l(k),r(k)]$, whose radius is $\rho(k)=(r(k)-l(k))/2$. 

The next theorem provides a solution of problem \eqref{eq:thsel}.

\begin{theorem}
\label{th:thm1}
Consider system \eqref{eq:sys1d1}-\eqref{eq:sys1d2} with the quantized sensor \eqref{eq:quant}. If $a\rho(k-1)+\delta_w > \delta_v$, a solution of problem \eqref{eq:thsel}  is given by
\begin{eqnarray}
\tau_c(k)&=& \ds a\, \frac{l(k-1)+r(k-1)}{2}, \label{eq:Cck}\\
\Delta(k)&=& 
\frac{2}{d+1}\left( a\rho(k-1)+\delta_w-\delta_v\right). \label{eq:Deltak}
\end{eqnarray}
When $a\rho(k-1)+\delta_w \leq \delta_v$, any $\theta(k)$ with $\tau_c(k)$ given by \eqref{eq:Cck}  and $\Delta(k)>0$ is an optimal solution of \eqref{eq:thsel}.
\end{theorem}
\begin{proof}
By applying \eqref{eq:sys1d1} to $\Xi(k-1|k-1)$, one has 
\begin{equation}
\begin{array}{rcl}
\Xi(k|k-1)&=&[a l(k-1)-\delta_w, a r(k-1)+\delta_w]=\\
&=& [l(k|k-1),r(k|k-1)].
\end{array}
\label{eq:predXi}
\end{equation} 
Problem \eqref{eq:thsel} requires to select $\theta(k)$ in order to minimize the worst-case radius of the interval $\Xi(k|k)=\Xi(k|k-1) \bigcap \cM(k)$, for all possible values taken by the quantized output signal $y(k)$. 
By using \eqref{eq:cMk} one has
\begin{equation}
\Xi(k|k) =
\left\{
\begin{array}{l}
[l(k|k-1),\min\{\tau_1+\delta_v,r(k|k-1)\}] \hspace*{2cm}~\mbox{ if } y(k)=0;\\   
{[}\max\{l(k|k-1),\tau_i-\delta_v\},\min\{\tau_{i+1}+\delta_v,r(k|k-1)\}] \\
\hspace*{5.8cm} \mbox{ if } y(k)=i, ~1\leq i \leq d-1; \vspace*{-2mm}\\~
\vdots \\
{[}\max\{l(k|k-1),\tau_d-\delta_v\},r(k|k-1)]   \hspace*{2cm}~\mbox{ if } y(k)=d .
\end{array}
\right.
\label{eq:Xikk1}
\end{equation}
By adopting a simple geometric argument, the thresholds of the quantizer should be symmetrically distributed over $\Xi(k|k-1)$. Therefore, the central threshold $\tau_c$ must coincide with the center of $\Xi(k|k-1)$, from which \eqref{eq:Cck} follows. 
Now, substituting \eqref{eq:thresh1} and \eqref{eq:Cck} into \eqref{eq:Xikk1}, after some manipulations one obtains 
\begin{equation}
\rho(k) =
\left\{
\begin{array}{l}
\min\left\{a\rho(k-1)+\delta_w,~ \frac{1}{2} \left(a \rho(k-1)+\delta_w+\delta_v-\frac{d-1}{2}\Delta(k) \right)\right \} \\
\hspace*{8cm}~\mbox{ if } y(k)\in\{0,d\};\\   
\min\left\{\frac{1}{2}\Delta(k) +\delta_v,~ a\rho(k-1)+\delta_w, \right. \vspace*{1mm}\\ 
\hspace{0.9cm} \left. \frac{1}{2} \left[a \rho(k-1) +\delta_w+\delta_v -\left( \left| \frac{d}{2}-y(k) \right| -\frac{1}{2} \right) \Delta(k) \right]\right\} \\
\hspace*{6.7cm}~\mbox{ if } y(k)\in\{1,\dots,d-1\}.\\   
\end{array}
\right.
\label{eq:rhok1}
\end{equation}
If $a\rho(k-1)+\delta_w \leq \delta_v$, for $y(k)=\frac{d-1}{2}$ one has $\rho(k)=a \rho(k-1) +\delta_w$. It is easy to check that for all other outcomes of $y(k)$, smaller values of $\rho(k)$ are obtained in \eqref{eq:rhok1}, hence
\begin{equation}
\sup_{y(k)} \rho(k) = a \rho(k-1) +\delta_w.
\label{eq:suprhok1}
\end{equation}
Since the expression in \eqref{eq:suprhok1} does not depend on $\Delta(k)$, any value of this parameter is optimal.
\\
Conversely, if $a\rho(k-1)+\delta_w \geq \delta_v$, from \eqref{eq:rhok1} one gets
\begin{equation}
\rho(k) =
\left\{
\begin{array}{l}
\frac{1}{2} \left(a \rho(k-1)+\delta_w+\delta_v -\frac{d-1}{2}\Delta(k) \right) 
\hspace*{2.1cm}~\mbox{ if } y(k)\in\{0,d\};
\vspace*{2mm}\\   
\min\left\{\frac{1}{2}\Delta(k) +\delta_v, \right. \vspace*{1mm}\\ 
\hspace{0.9cm} \left. \frac{1}{2} \left[a \rho(k-1) +\delta_w+\delta_v -\left( \left| \frac{d}{2}-y(k) \right| -\frac{1}{2} \right) \Delta(k) \right]\right\} \\
\hspace*{6.7cm}~\mbox{ if } y(k)\in\{1,\dots,d-1\}.\\   
\end{array}
\right.
\end{equation}
By applying some straightforward manipulations, one obtains
\begin{equation}
\sup_{y(k)} \rho(k) =
\left\{
\begin{array}{l}
\frac{1}{2} \left(a \rho(k-1)-\frac{d-1}{2}\Delta(k) +\delta_w+\delta_v\right) \vspace*{1mm}\\
\hspace*{4.0cm}~\mbox{ if } \Delta(k) \leq \frac{2}{d+1}\left( a\rho(k-1)+\delta_w-\delta_v\right) \vspace*{2mm} \\
\frac{1}{2}\Delta(k) +\delta_v \\
\hspace*{0.2cm}~\mbox{ if }  \frac{2}{d+1}\left( a\rho(k-1)+\delta_w-\delta_v\right) \leq \Delta(k) \leq  a\rho(k-1)+\delta_w-\delta_v \vspace*{2mm}\\
\frac{1}{2} \left(a \rho(k-1) +\delta_w+\delta_v\right) \vspace*{1mm}\\
\hspace*{5.0cm}~\mbox{ if }  \Delta(k) \geq  a\rho(k-1)+\delta_w-\delta_v
\end{array}
\right.
\label{eq:suprhok2}
\end{equation}
whose minimum with respect to the the quantizer resolution $\Delta(k)$ is attained for $\Delta(k) = \frac{2}{d+1} ( a\rho(k-1)+\delta_w-\delta_v )$, which corresponds to \eqref{eq:Deltak}.
\end{proof}

The rational behind the quantizer adaptation \eqref{eq:Cck}-\eqref{eq:Deltak} is clarified by the following example.
\begin{example}
Let $a=2$, $\delta_w=1$, $\delta_v=2$. Assume that at time $k\!-\!1$ one has $\Xi(k-1|k-1)=[-5,~5]$. Then, from \eqref{eq:predXi} one gets $\Xi(k|k-1)=[-11,~11]$. By applying \eqref{eq:Cck}-\eqref{eq:Deltak} with $d=5$, one has $\tau_c(k)=0$ and $\Delta(k)=3$. The resulting thresholds $\tau_i(k)$ in \eqref{eq:thresh1} are equal to $-6,-3,0,3,6$. It is worth stressing that such thresholds do not partition the predicted feasible set $\Xi(k|k-1)$ in equal parts. Rather, they are selected in such a way that the corrected feasible set $\Xi(k|k)=\Xi(k|k-1) \bigcap \cM(k)$ has always the same radius, equal to $\frac{1}{2}\Delta(k)+\delta_v=\frac{7}{2}$, irrespectively of the sensor output $y(k)$, thus achieving the optimal solution of problem \eqref{eq:thsel}.
\end{example}

By using the result in Theorem \ref{th:thm1} it is possible to establish conditions for the asymptotic boundedness of the feasible set $\Xi(k|k)$. 
\begin{theorem}
\label{th:thm2}
By choosing the quantizer parameters as in \eqref{eq:Cck}-\eqref{eq:Deltak}, the radius of the feasible set $\Xi(k|k)$ evolves according to
\begin{equation}
\rho(k)=\left\{
\begin{array}{ll}
a \rho(k-1) +\delta_w & \mbox{ if } \rho(k-1) \leq \frac{\delta_v-\delta_w}{a}, \vspace*{2mm}\\
\frac{1}{d+1} \left(a \rho(k-1) +\delta_w+d\delta_v \right) & \mbox{ if } \rho(k-1) \geq \frac{\delta_v-\delta_w}{a}.
\end{array}
\right.
\label{eq:rhok2}
\end{equation}
Moreover, if $d>a-1$, one has
\begin{equation}
\lim_{k \rightarrow +\infty} \rho(k)=\left\{
\begin{array}{ll}
\frac{1}{1-a} \delta_w & \mbox{ if } \delta_w \leq (1-a) \delta_v, \vspace*{2mm}\\
\frac{1}{d+1-a} \left(\delta_w+d \delta_v \right) & \mbox{ else,}
\end{array}
\right.
\label{eq:limrhok}
\end{equation}
and if  $\delta_w > (1-a) \delta_v$,
\begin{equation}
\lim_{k \rightarrow +\infty} \Delta(k)=
\frac{2}{d+1-a} \left(\delta_w+(a-1) \delta_v \right) .
\label{eq:limDelta}
\end{equation}
\end{theorem}
\begin{proof}
Expression \eqref{eq:rhok2} is obtained from \eqref{eq:suprhok1} and \eqref{eq:suprhok2}, by substituting \eqref{eq:Deltak} into \eqref{eq:suprhok2}. Then, \eqref{eq:limrhok} follows from standard arguments on the asymptotic behavior of system \eqref{eq:rhok2}. Finally, \eqref{eq:limDelta} is obtained by substituting the second expression in \eqref{eq:limrhok} into \eqref{eq:Deltak} (note that when the first expression holds, the choice of $\Delta(k)$ becomes asymptotically irrelevant).
\end{proof}

Theorem \ref{th:thm2} states that in order to guarantee asymptotic boundedness of the feasible state set, the number of thresholds of a uniform quantizer must exceed $a-1$. Notice that if system \eqref{eq:sys1d1} is stable, this condition is always satisfied. Conversely, in the special case of a binary sensor ($d=1$), asymptotic boundedness holds if and only if $a<2$.

\begin{remark} 
If $a<0$, Theorem \ref{th:thm2} still holds by replacing $a$ with $|a|$ in \eqref{eq:rhok2}-\eqref{eq:limDelta}.
If $d$ is even, the results in Theorems \ref{th:thm1} and \ref{th:thm2} do not change, except that there will be no threshold coinciding with the center of $\Xi(k|k-1)$. However, the optimal resolution $\Delta$ in \eqref{eq:Deltak} clearly remains the same. The optimal threshold locations (for any $d$) are obtained by substituting \eqref{eq:Cck}-\eqref{eq:Deltak} in \eqref{eq:thresh1}, thus obtaining
\begin{equation}
\tau_i(k)=\ds a\, \frac{l(k-1)+r(k-1)}{2} - \left(1- \frac{2i}{d+1} \right) \left( a\rho(k-1)+\delta_w-\delta_v\right),
\label{eq:thresh2}
\end{equation}
for $i=1,\dots,d$.
\end{remark}

\begin{remark} 
The result \eqref{eq:limrhok} in Theorem \ref{th:thm2} can be seen as the set-theoretic counterpart of that in \cite[Th.~3]{wb97} for a stochastic estimation framework. In fact, the cited result provides the minimum number $n$ of bits for a coder-estimator guaranteeing asymptotic boundedness of the covariance of the state estimation error. Specifically, under the assumption that all the involved stochastic disturbances have a finite-support density function, it is shown that $2^n > a$ implies asymptotic stability of a coder estimator for a first-order system with pole $a$. By observing that an $n$-bit coder corresponds to a quantizer with $d=2^n-1$ thresholds, the condition $d>a-1$ is recovered. In addition, Theorem \ref{th:thm2} provides exact asymptotic expressions for the estimation uncertainty $\rho(k)$ and the quantizer resolution $\Delta(k)$, in the considered set-theoretic framework.
\end{remark}

\section{General case}
\label{sec:nth_order}

The solution of problem \eqref{eq:thsel} in the general case of $n$-th order systems is intractable, because it involves a nonconvex min-max optimization over $n$-dimensional sets. Nevertheless, we can leverage the solution for first-order systems, in order to solve an alternative formulation of the problem, which is both meaningful and computationally feasible.

Instead of minimizing the radius of the feasible state set at time $k$, the aim is to minimize the radius of the \emph{feasible output signal set}, i.e., the set of feasible values that the nominal output $\bar{z}(k)=Cx(k)$ can take at time $k$. This is given by
\begin{equation}
C\Xi(k|k)=C\left(\Xi(k|k-1)\bigcap \cM(k) \right) = C\Xi(k|k-1)\bigcap C\cM(k) 
\label{eq:foss}
\end{equation}
where the last equality follows from the definition of $\cM(k)$ in \eqref{eq:cMk}. 
Then, one can design the quantizer parameters according to
\begin{equation}
\begin{array}{rcl}
\theta^*(k) &=&\ds \arg \inf_{\theta} \sup_{y(k)} \rad (C\Xi(k|k)) \\
&=& \ds \arg \inf_{\theta} \sup_{y(k)} \rad \left(C\Xi(k|k-1) \bigcap C\cM(k) \right).  
\end{array}
\label{eq:thseln}
\end{equation}
From a geometric viewpoint, the difference between problems \eqref{eq:thsel} and \eqref{eq:thseln} is that in the latter it is the projection of the uncertainty set onto the  output measurement manifold which is minimized. 

Since $C\Xi(k|k)$ in \eqref{eq:foss} is a one-dimensional interval, one can solve problem \eqref{eq:thseln} by adopting the same strategy as in the case of first-order systems. 
In particular, by using \eqref{eq:predXi}, the optimal quantizer parameters in Theorem \ref{th:thm1} can be rewritten as
\begin{eqnarray}
\tau_c(k)&=& \frac{l(k|k-1)+r(k|k-1)}{2}, \label{eq:Cck_th1}\\
\Delta(k)&=& \frac{r(k|k-1)-l(k|k-1)-2\delta_v}{d+1}, \label{eq:Deltak_th1}
\end{eqnarray}
where $l(k|k-1)$ and $r(k|k-1)$ are, respectively, the infimum and the supremum of the interval $\Xi(k|k-1)$, namely the predicted feasible state set.
Similarly to \eqref{eq:Cck_th1}-\eqref{eq:Deltak_th1}, for an $n$-th oder system one can choose the quantizer parameters as
\begin{eqnarray}
\tau_c(k)&=& \frac{\bar{l}(k)+\bar{r}(k)}{2}, \label{eq:Cckn}\\
\Delta(k)&=& \frac{\bar{r}(k)-\bar{l}(k)-2\delta_v}{d+1}, \label{eq:Deltakn}
\end{eqnarray}
where
\begin{equation}
\bar{l}(k)=\inf_{x\in\Xi(k|k-1)} Cx~,~~~ \bar{r}(k)=\sup_{x\in\Xi(k|k-1)} Cx.
\label{eq:barlbarr}
\end{equation}
In words, this choice corresponds to selecting the sensor thresholds in such a way that the feasible output signal set in \eqref{eq:foss} turns out to have the same radius, irrespectively of the value taken by the output $y(k)$, thus leading to the optimal solution of problem \eqref{eq:thseln}.

A question that naturally arises is whether the feasible state set is asymptotically bounded when the quantizer is adapted according to \eqref{eq:Cckn}-\eqref{eq:barlbarr}.
First, it can be observed that this property trivially holds when system \eqref{eq:sys1} is asymptotically stable. Indeed, even in the least informative scenario in which no measurement set $\cM(k)$ is able to reduce the size of the feasible set, namely $\cM(k) \supseteq \Xi(k|k-1)$, $\forall k \geq 0$, from \eqref{eq:fss1}-\eqref{eq:fss2} one has 
\begin{equation}
\label{eq:xinomeas}
\Xi(k+1|k+1) =  A \Xi(k|k) + \delta_wG \cB_\infty.
\end{equation}
Then, if $A$ is asymptotically stable, it turns out that $\Xi(k|k)$ is asymptotically included in the bounded set $\sum_{i=0}^{+\infty} A^i  \delta_wG \cB_\infty$ (see, e.g., \cite{kg98,bm08}).

When the stability assumption is not enforced, the following general result holds.

\begin{theorem}
\label{th:thresh}
Assume matrix $A$ is nonsingular and system \eqref{eq:sys1}-\eqref{eq:sys2} is observable. Define the matrix
$$
\Omega=\left[ \begin{array}{l} C \\ CA^{-1} \\ \vdots \\ CA^{-n+1} \end{array} \right]
$$
and the vector $\alpha=CA\,\Omega^{-1}$. Let $\eta^*>0$ be such that the polynomial 
\begin{equation}
q(z)=z^n-\eta\sum_{i=1}^n |\alpha_i|z^{n-i}
\label{eq:polyz}
\end{equation}
has all the roots strictly inside the unit circle, for all $\eta\in[0,\eta^*)$.
Then, if the quantizer parameters are chosen according to \eqref{eq:Cckn}-\eqref{eq:barlbarr}, with $d\in\nn$ such that
\begin{equation}
d > \frac{1}{\eta^*} -1
\label{eq:condp}
\end{equation}
then the feasible state set $\Xi(k|k)$ is asymptotically bounded, i.e., there exists $\rho_\infty>0$ such that $\limsup_{k \rightarrow +\infty} \rho(k) \leq \rho_\infty$.
Moreover, for every $d$ satisfying \eqref{eq:condp}, one has
\begin{equation}
\rho_\infty \leq \max_{w\in\cB_\infty} \| \Omega^{-1}D_\infty w\|
\label{eq:rhoinf}
\end{equation}
where 
\begin{equation}
D_\infty=\mbox{diag}\left(\left[\frac{1}{2} \Delta_\infty + \delta_v \right]  \mathds{1} +\delta_w \beta \right)
\label{eq:Dinf}
\end{equation}
with
\begin{equation}
\Delta_\infty = \frac{2}{d+1-\|\alpha\|_1} \left\{ \delta_v\left(\|\alpha\|_1-1\right) + \delta_w \left( \|CG\|_1+ \sum_{h=1}^n |\alpha_h| \beta_h \right) \right\}
\label{eq:Deltainf}
\end{equation}
and $\beta\in\rr^n$ is defined as
\begin{equation}
\beta_1=0~,~~~\beta_h = \sum_{i=0}^{h-2} \| CA^{-(i+1)} G\|_1~,~~h=2,\dots,n
\label{eq:betah}
\end{equation}
\end{theorem}
\begin{proof}
See Appendix \ref{app:th_thresh}.
\end{proof}

Theorem \ref{th:thresh} provides an upper bound to the minimum number of thresholds for which asymptotic boundedness of the feasible state set is guaranteed. Such a bound is obtained by computing
$\eta^*=\sup\{\eta:~q(z) \mbox{~is Schur}\}$ and then choosing $d$ as the minimum integer strictly larger than $\frac{1}{\eta^*} -1$.
In general, the bound \eqref{eq:condp} may be conservative. However, there are cases in which it turns out to be quite tight, as in the following numerical examples.

\begin{example}
Consider system \eqref{eq:sys1}-\eqref{eq:sys2} with
$$
A=\left[ \begin{array}{cc} a & 1   \\ 0 & a \end{array} \right],~~~C=[1 ~~0 ],~~~G=I_2
$$
with $a \in \rr$, $\delta_w=0.05$ and $\delta_v=0$. The quantizer parameters are updated according to \eqref{eq:Cckn}-\eqref{eq:barlbarr}, for different values of $d$. Figure \ref{fig:ex2Dp} reports, for different values of $a$: (i) the minimum $d \in \nn$ for which the feasible set $\Xi(k|k)$, resulting from the worst-case realization of the quantized output $y(k)$, remains asymptotically bounded; (ii) the minimum $d \in \nn$ satisfying \eqref{eq:condp}. 
It can be observed that the upper bound given by Theorem \ref{th:thresh} is only slightly larger than the minimum number of thresholds for which the feasible state is actually bounded. In this example, the upper bound is tight for $a \geq 2.3$. 
\label{ex1}
\end{example}

\begin{figure}[pht]
   \centering
   \includegraphics[width=.7\columnwidth]{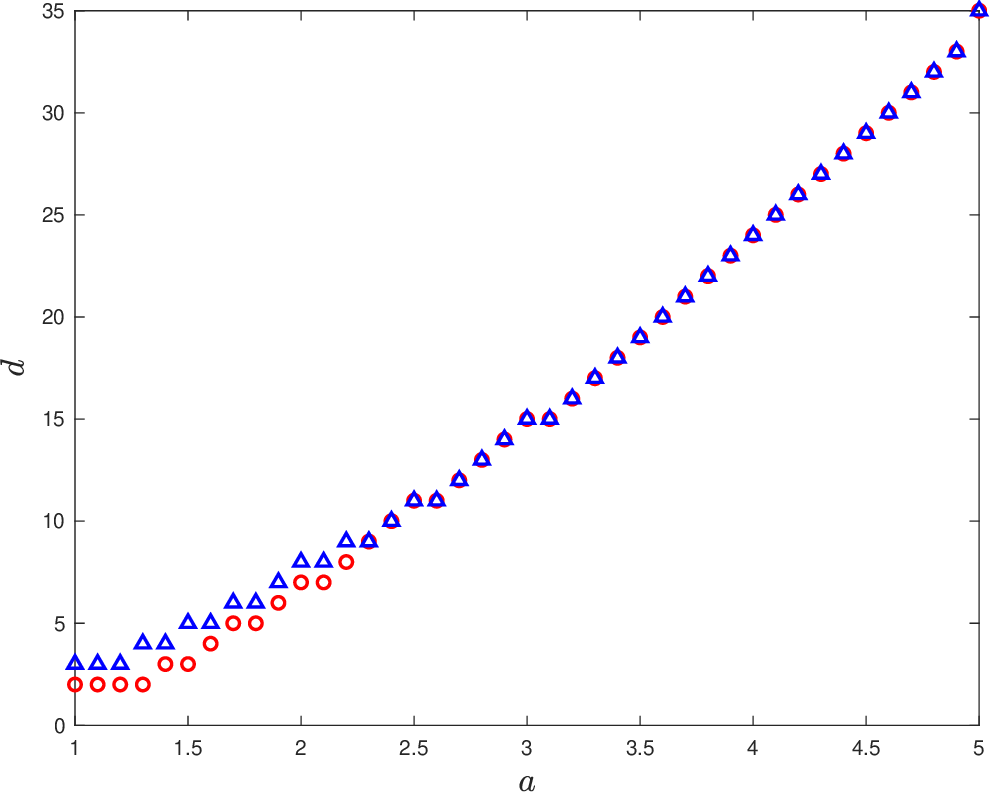}
   \caption{Example \ref{ex1}: minimum $d$ for which  the worst-case $\Xi(k|k)$ is asymptotically bounded (red circles) and minimum $d$ satisfying \eqref{eq:condp} (blue triangles), for different values of $a$.}
      \label{fig:ex2Dp}
 \end{figure}

\section{Threshold selection based on set approximations}
\label{sec:approx}

The main drawback of the threshold selection mechanism \eqref{eq:Cckn}-\eqref{eq:barlbarr} is that it requires the propagation of the true feasible state set $\Xi(k|k)$ according to the recursion \eqref{eq:fss1}-\eqref{eq:fss2}. Unfortunately, this is not computationally feasible, even for low state dimensions $n$. This is a well-known problem in the set membership estimation literature, which led to the development of a number of  set approximation techniques. 
A popular approach is based on the choice of a class of approximating sets $\cAS$, from which an outer approximation of the true feasible set is selected at each time $k$ through the recursion
\begin{eqnarray}
\cAS(k|k) &\supseteq& \cAS(k|k-1) \bigcap \cM(k) \label{eq:cS1}\\
\cAS(k+1|k) &\supseteq& A \cAS(k|k) + \delta_wG\cB_\infty \label{eq:cS2}
\end{eqnarray}
which is initialized by picking a $ \cAS(0|-1) \supseteq \Xi(0|-1)$.
Many different set classes have been considered in the literature, including ellipsoids \cite{gnkh96,dwp01,ec01}, parallelotopes \cite{CGZ96}, zonotopes \cite{abc05,Com15, wwpc19,rrsr20} and constrained zonotopes \cite{srmb16,rk22,rlrr22}.

When an outer approximation of the feasible set is available, one can design the quantizer thresholds by applying the ideas developed in Section \ref{sec:nth_order} to the approximating set $\cAS(k|k-1)$.
In particular,  $\tau_c(k)$ and $\Delta(k)$ can still be chosen according to \eqref{eq:Cckn}-\eqref{eq:Deltakn}, where $\bar{l}(k)$ and $\bar{r}(k)$ are computed as
\begin{equation}
\bar{l}(k)=\inf_{x\in\cAS(k|k-1)} Cx~,~~~ \bar{r}(k)=\sup_{x\in\cAS(k|k-1)} Cx.
\label{eq:barlbarrcS}
\end{equation}
Notice that $\bar{l}(k)$ and $\bar{r}(k)$ are easy to compute when the approximating sets $\cAS$ are parallelotopes or zonotopes. If $x_c(k|k-1)$ and $T(k|k-1)$ denote respectively the center and the generator matrix associated to $\cAS(k|k-1)$, according to the notation in \eqref{eq:parallelotope} and \eqref{eq:zonotope}, one gets
\begin{equation}
\bar{l}(k)=Cx_c(k|k-1)-\|CT(k|k-1)\|_1~,~~~ \bar{r}(k)=Cx_c(k|k-1)+\|CT(k|k-1)\|_1
\label{eq:barlbarrcS2}
\end{equation}
and hence \eqref{eq:Cckn}-\eqref{eq:Deltakn} boil down to
\begin{eqnarray}
\tau_c(k) &=& Cx_c(k|k-1), \label{eq:CcknZ} \\
\Delta(k) &=& \frac{2}{d+1} (\|CT(k|k-1)\|_1-\delta_v), \label{eq:DeltaknZ}
\end{eqnarray}
which can be easily computed at each time step.
When constrained zonotopes \eqref{eq:constrzonotope} are used as approximating sets, the computation of $\bar{l}(k)$ and $\bar{r}(k)$ in \eqref{eq:barlbarrcS} requires to solve two linear programs.

\section{Numerical examples}
\label{sec:numex}

In this section, the effectiveness of the threshold adaptation based on \eqref{eq:CcknZ}-\eqref{eq:DeltaknZ} is demonstrated on numerical examples.

\subsection{Case study: double oscillator}
\label{subsec:2osc}

We consider the double-spring mechanical system proposed in \cite{bcg17}, whose continuous-time equations are
$$
\begin{array}{rcl}
\dot{x}(t) &=& 
\left[ \begin{array}{rrrr} 
0 & 1 & 0 & 0 \\
-20 & 0 & 10 & 0\\
0 & 0 & 0 & 1 \\
10 & 0 & -10 & 0
\end{array}
\right]
x(t) 
+
\left[ \begin{array}{cc} 
0  & 0\\ 1 & 0 \\0 & 0 \\0 & 1
\end{array}
\right]
w(t)
\label{eq:sys1ex2} \vspace*{2mm}\\
z(t)&=&\left[ \begin{array}{cccc} 0 & 0 & 1 & 0 \end{array} \right]  \, x(t) + v(t) \label{eq:sys2ex2}
\end{array}
$$
The system is discretized with sampling time $0.1\,s$. The discrete-time disturbances $w(k)$ and $v(k)$ are generated as uniformly distributed white processes, satisfying \eqref{eq:ubbw} and \eqref{eq:ubbv}, with $\delta_w=0.2$ and $\delta_v=0.05$. The initial state is generated randomly within a box of side 5, i.e., $x(0) \in 5 \cB_\infty$.

In order to test the threshold selection mechanism based on outer approximation of the feasible state sets, illustrated in Section \ref{sec:approx}, we consider the following recursive set approximation algorithms.
\begin{itemize}
\item $Par$, that uses the recursive parallelotopic approximation proposed in \cite{CGZ96}, in which the intersection step \eqref{eq:cS1} is performed according to the minimum volume criterion devised in \cite{VZ96}.
\item $Zon(o_z)$, employing zonotopes with maximum order equal to $o_z$. The recursive approximation procedure is the one proposed in \cite{abc05}. The intersection step in \eqref{eq:cS1} is performed according to the minimum volume criterion illustrated in \cite{BAC06}. If the number of generators exceeds $o_zn$, the order reduction procedure proposed in \cite{srmb16} is applied.
\item $C\!Z(o_c,n_c)$, using constrained zonotopes with maximum order $o_c$ and maximum number of constraints $n_c$. The recursive approximation procedure and the order reduction technique are those proposed in \cite{srmb16}. The intersection step in \eqref{eq:cS1} is based on the minimization of the Frobenius norm of the zonotope generator matrix.
\end{itemize}
All operations involving zonotopes and constrained zonotopes have been performed by using the CORA toolbox for Matlab \cite{AltCORA21}.

A first set of experiments is carried out using the algorithm $Par$,  with $d=5$ thresholds. Figures \ref{fig:par_sel}-\ref{fig:par_nosel} show the state estimates provided by the center of the parallelotope (dahed) and the uncertainty interval associated to each state variable (solid). The latter has been obtained by computing the minimum box containing the current approximating parallelotope. The evolution of the true state  variables is shown in red. In Fig.~\ref{fig:par_sel}, the results obtained by adapting the quantizer thresholds according to \eqref{eq:CcknZ}-\eqref{eq:DeltaknZ} are reported. For comparison, Fig.~\ref{fig:par_nosel} shows the same quantities when the quantizer thresholds are kept constant and uniformly distributed in the interval $[-5,5]$ (which is approximately the range of the output signal $z(k)$). It can be clearly seen that the threshold adaptation mechanism is able to significantly reduce the uncertainty associated to the estimates, even for a relatively small number of thresholds.

\begin{figure}[p]
   \centering
   \includegraphics[width=.78\columnwidth]{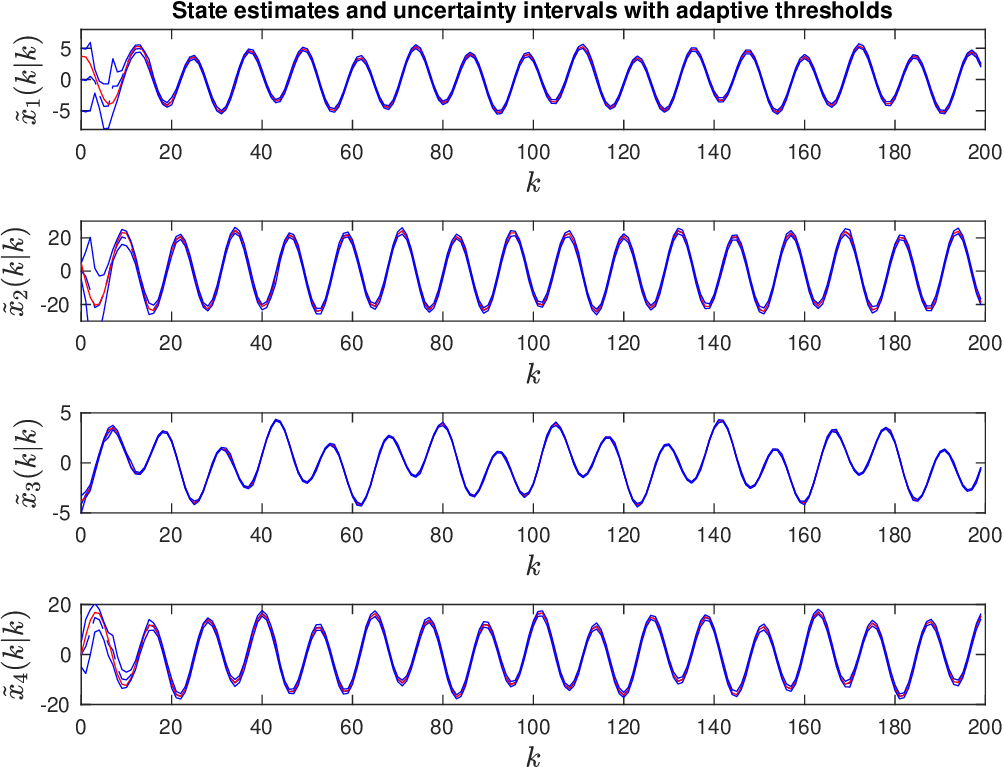}
   \caption{State estimates (dashed) and uncertainty bounds (solid) for the $Par$ algorithm with $d=5$ adaptive thresholds. True states are in red.}
      \label{fig:par_sel}
 \end{figure} 

\begin{figure}[p]
   \centering
   \includegraphics[width=.78\columnwidth]{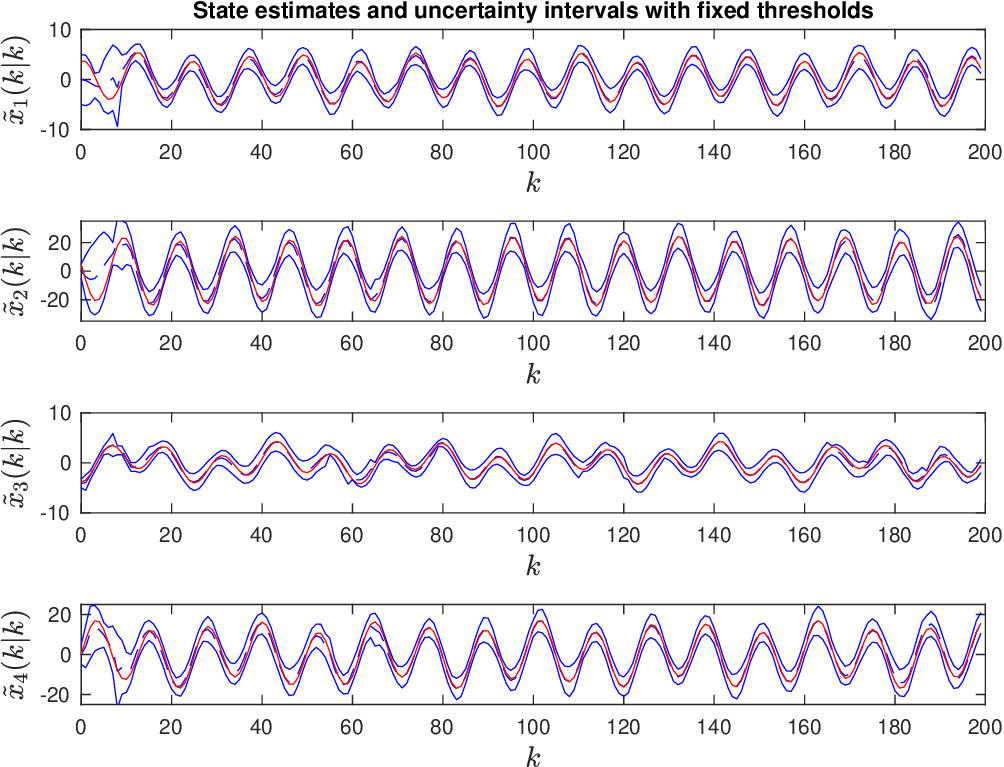}
   \caption{State estimates (dashed) and uncertainty bounds (solid) for the $Par$ algorithm with $d=5$ fixed thresholds. True states are in red.}
         \label{fig:par_nosel}
 \end{figure} 

In order to evaluate the role of the threshold number $d$, the volumes of the approximating parallelotopes and zonotopes are depicted in Fig.~\ref{fig:avg_vol} as a function of $d$. Results are averaged over 10 runs. As expected, the zonotopic approximation provides smaller volumes with respect to the parallelotopic one. The advantage saturates for zonotopes of order $4$, which have 16 generators in $\rr^4$.

\begin{figure}[t]
   \centering
   \includegraphics[width=.78\columnwidth]{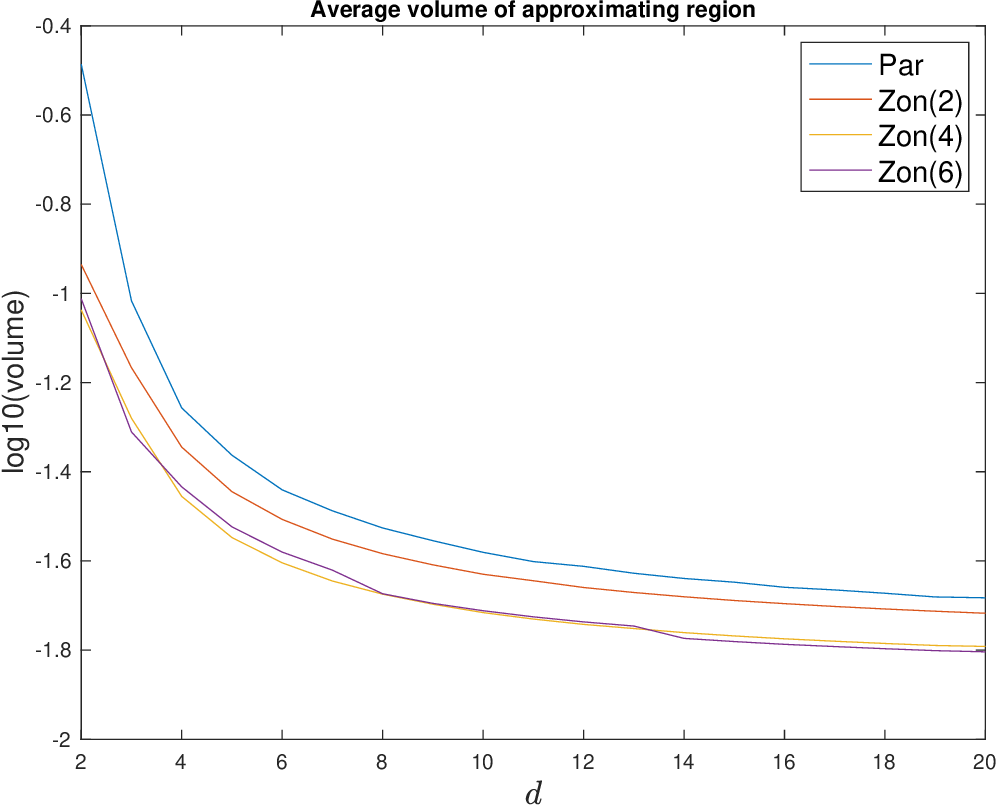}
   \caption{Log-volumes of approximating regions for different number of thresholds $d$.}
      \label{fig:avg_vol}
 \end{figure}

To compare the performance of the different set approximation techniques, we consider the average size of the uncertainty associated to the state estimates. Namely, the radius of the uncertainty interval associated to $x_i(k)$ is averaged over time $k$ and over the state variables $i=1,\dots,4$, for 10 runs involving different initial conditions and disturbance realizations. The time average is performed over the interval $k \in [50,200]$ to remove the effect of the initial transient. The results concerning the considered approximation techniques are reported in Tables \ref{tab:sel}-\ref{tab:nosel}, for different values of $d$. Table \ref{tab:sel} refers to the thresholds adapted according to \eqref{eq:CcknZ}-\eqref{eq:DeltaknZ}, while Table \ref{tab:nosel} concerns the case of fixed thresholds. 
Table \ref{tab:times} reports the relative times for one iteration of each approximation algorithm, with respect to one iteration of $Par$. The average time for one iteration of $Par$ was $0.16 \, ms$ on a computer equipped with Apple M2 CPU and 24 GB of RAM. Clearly, the computational burden of $Zon$ and $CZ$ algorithms can be triggered by acting on the zonotope order and, for $CZ$, on the number of equality constraints, which in the considered tests has been kept equal to 10.

The benefit of the threshold selection mechanism is apparent, not only in terms of uncertainty size, but also for the number of thresholds. Indeed, uncertainty intervals do not significantly decrease for $d>10$ when thresholds are adapted, while they keep decreasing (also for $d>20$) when thresholds are fixed.
Concerning the different techniques, constrained zonotopes yield smaller uncertainty intervals, as expected, at the price of a much higher computational burden. Uncertainties provided by the parallelotopic and zonotopic approximations are quite close, with the former approach providing a good trade-off between quality of the estimates and required computational effort.

\begin{table}[p]
\begin{center}
\begin{tabular}{|c|c|c|c|c|c|c|c|}
\hline 
$d$ & $Par$ & $Zon(2)$ &$Zon(4)$ & $Zon(6)$ & $C\!Z(2,10)$ & $C\!Z(4,10)$ & $C\!Z(6,10)$ \\ 
\hline 
3 & 0.93 & 0.78 & 0.66 & 0.62 & 0.56 & 0.5 & 0.46 \\ 
5 & 0.76 & 0.66 & 0.58 & 0.54 & 0.5 & 0.43 & 0.41 \\ 
10 & 0.69 & 0.59 & 0.52 & 0.49 & 0.41 & 0.35 & 0.32 \\ 
15 & 0.67 & 0.57 & 0.50 & 0.48 & 0.38 & 0.32 & 0.30 \\ 
20 & 0.66 & 0.56 & 0.50 & 0.47 & 0.36 & 0.31 & 0.29 \\ 
\hline 
\end{tabular}
\end{center}
\caption{Average uncertainty associated to the state estimates for different $d$: adaptive thresholds.}
\label{tab:sel}
\end{table}

\begin{table}[p]
\begin{center}
\begin{tabular}{|c|c|c|c|c|c|c|c|}
\hline 
$d$ & $Par$ & $Zon(2)$ &$Zon(4)$ & $Zon(6)$ & $C\!Z(2,10)$ & $C\!Z(4,10)$ & $C\!Z(6,10)$ \\ 
\hline 
3 & 5.68 & 7.55 & 6.41 & 6.37 & 2.43 & 2.28 & 2.24 \\ 
5 & 4.54 & 4.96 & 5.01 & 4.94 & 1.93 & 1.81 & 1.75 \\ 
10 & 2.87 & 3.08 & 3.01 & 2.93 & 1.46 & 1.39 & 1.35 \\ 
15 & 2.27 & 2.35 & 2.21 & 2.15 & 1.19 & 1.11 & 1.07 \\ 
20 & 1.93 & 1.92 & 1.84 & 1.77 & 1.05 & 0.98 & 0.96 \\ 
\hline 
\end{tabular}
\end{center}
\caption{Average uncertainty associated to the state estimates for different $d$: fixed thresholds.}
\label{tab:nosel}
\end{table}

\begin{table}[p]
\begin{center}
\begin{tabular}{|c|c|c|c|c|c|c|}
\hline 
$Par$ & $Zon(2)$ &$Zon(4)$ & $Zon(6)$ & $C\!Z(2,10)$ & $C\!Z(4,10)$ & $C\!Z(6,10)$ \\ 
\hline 
1 & 2.9 & 17.8 & 90.5 & 1879.9 & 3313.8 & 4858.5 \\ 
\hline 
\end{tabular}
\end{center}
\caption{Average relative times for one iteration of each approximation algorithm, with respect to one iteration of $Par$.}
\label{tab:times}
\end{table}

\subsection{Randomly generated multivariable systems}

The adaptive threshold selection procedure has been tested on randomly generated multi-input multi-output linear systems, with $n=5$ states.
The systems are asymptotically stable and fed with $m=3$ input signals $w(k)$ which are generated as sum of sinusoids with randomly chosen frequencies, corrupted by uniformly distributed white process disturbances satisfying \eqref{eq:ubbw} with $\delta_w=0.1$.
Quantized measurements of $3$ output signals $z_i(k)=C_i x(k) +v_i(k)$, $i=1,2,3$ are collected, with uniformly distributed white noise $v_i(k)$ satisfying \eqref{eq:ubbv} with $\delta_v=0.1$.
For each quantized measurement, adaptive thresholds have been selected according to the procedure proposed  in Section \ref{sec:approx}, based  on outer approximation of the feasible state sets. 
In all recursive set approximation algorithms, quantized measurements are processed sequentially, i.e., step \eqref{eq:cS1} is repeated for each feasible measurement set $\cM_i(k)$ corresponding to the quantized measurement $y_i(k)=\sigma(z_i(k))$. 

Tables \ref{tab:ms_sel}-\ref{tab:ms_nosel} report the results of simulations performed on 100 randomly generated systems, each one lasting 100 time instants. The radius of the uncertainty interval associated to state estimates is shown, for different number of thresholds $d$. Results are averaged with respect to time, state variables and tested systems. The employed set approximation techniques are the same as in Section \ref{subsec:2osc}. The order of zonotopes has been limited to $o=4$ because negligible benefits are obtained by increasing such value. For similar reasons $n_c=5$ constraints have been used for constrained zonotopes.
 
Table \ref{tab:ms_sel} refers to the case of adaptive thresholds, while those in Table \ref{tab:ms_nosel} concern a scenario in which thresholds are kept constant. In order to reduce the uncertainty in  the latter scenario, for each randomly generated system the thresholds have been chosen equally spaced in the exact range of the unaccessible output signal $z(k)$. Nevertheless, it can be seen that adaptive thresholds outperform constant ones. On average, the size of state uncertainty is about three times smaller when thresholds are tuned online with respect to when they are fixed. Notice that this happens almost irrespectively of both the number of thresholds and the adopted set approximation technique. The average relative times taken by each algorithm with respect to the $Par$ one are reported in Table \ref{tab:ms_times}. It is confirmed that constrained zonotopes provide the smallest uncertainty, but they are more suitable to offline operations due to the high computational burden. For this study, zonotopes of order $2$ seem to be the best compromise between quality of the state estimates and required  computational effort.
\begin{table}[p]
\begin{center}
\begin{tabular}{|c|c|c|c|c|c|}
\hline 
d & Par & Zon(2) & Zon(4) & CZ(2,5) & CZ(4,5) \\ 
\hline 
3 & 1.44 & 0.75 & 0.77 & 0.44 & 0.42 \\ 
5 & 0.97 & 0.54 & 0.51 & 0.36 & 0.34 \\ 
10 & 0.62 & 0.41 & 0.39 & 0.29 & 0.27 \\ 
15 & 0.53 & 0.37 & 0.35 & 0.27 & 0.25 \\ 
20 & 0.48 & 0.35 & 0.34 & 0.26 & 0.24 \\ 
\hline 
\end{tabular}
\end{center}
\caption{Randomly generated systems: average uncertainty associated to the state estimates for different $d$: adaptive thresholds.}
\label{tab:ms_sel}
\end{table}
\begin{table}[p]
\begin{center}
\begin{tabular}{|c|c|c|c|c|c|}
\hline 
d & Par & Zon(2) & Zon(4) & CZ(2,5) & CZ(4,5) \\ 
\hline 
3 & 5.52 & 1.98 & 1.36 & 1.09 & 0.87 \\ 
5 & 3.08 & 1.62 & 1.22 & 0.97 & 0.81 \\ 
10 & 2.01 & 1.26 & 1.03 & 0.82 & 0.71 \\ 
15 & 1.56 & 1.07 & 0.93 & 0.72 & 0.65 \\ 
20 & 1.35 & 0.97 & 0.87 & 0.67 & 0.61 \\ 
\hline 
\end{tabular}
\end{center}
\caption{Randomly generated systems: average uncertainty associated to the state estimates for different $d$: fixed thresholds.}
\label{tab:ms_nosel}
\end{table}
\begin{table}[p]
\begin{center}
\begin{tabular}{|c|c|c|c|c|}
\hline 
Par & Zon(2) & Zon(4) & CZ(2,5) & CZ(4,5) \\ 
\hline 
1 & 2.9 & 4.8 & 2042.4 & 3521.8 \\ 
\hline 
\end{tabular}
\end{center}
\caption{Randomly generated systems: average relative times for one iteration of approximation algorithms, with respect to one iteration of $Par$.}
\label{tab:ms_times}
\end{table}

\section{Conclusions}
\label{sec:conc}

A new approach to the selection of adaptive thresholds for state estimation with quantized measurements has been proposed. The method relies on the set membership estimation paradigm: thresholds are chosen in such a way to minimize the size of the feasible state set, and hence the uncertainty affecting the state estimates.
It has been shown that a computationally cheap solution can be derived by exploiting set approximations based on parallelotopes or zonotopes. Numerical results show that remarkable uncertainty reductions are achieved with respect of quantizers with fixed thresholds.
Future developments of this research may concern more accurate quantization schemes for multi-output systems, taking into account the geometry of the feasible measurement set. The application of the proposed approach to nonlinear systems will be also investigated.

\appendix
\section{Proof of Proposition \ref{prop1}}
\label{app:prop1}
Let $y\in\cS(p,c)$ and $w\in \rr^m$ such that $w\in\cB_\infty$. Set $x=Ay+\gamma G w$. Then
$$
\begin{array}{l}
|p'A^{-1}x-c|=|p'y-c+\gamma p'A^{-1}G w| \vspace*{1mm}\\
\leq |p'y-c|+ \ds \gamma \sup_{w\in\cB_\infty} |p'A^{-1}G w| \leq 1 + \gamma \|p'A^{-1}G\|_1.
\end{array}
$$
Conversely, let $x$ be such that 
\begin{equation}
|p'A^{-1}x-c |\leq 1 + \gamma \|p'A^{-1}G\|_1.
\label{eq:extstr}
\end{equation}
Set  $y=A^{-1}(x-\gamma G \bar{w})$, for some $\bar{w}$. We need to show that it is always possible to choose $\bar{w}\in\cB_\infty$ so that $y\in\cS(p,c)$, and hence $x=Ay+\gamma G \bar{w} \in A\cS(p,c)+\gamma G \cB_\infty$. Consider the case $p'A^{-1}x-c \geq 0$. Then choose $\bar{w}$ such that
$$
\bar{w}_i=\mbox{sign}\{(p'A^{-1}G)_i\} \frac{\max\{p'A^{-1}x-c-1,0\}}{\gamma \|p'A^{-1}G\|_1}.
$$
for $i=1,\dots,m$, with $(v)_i$ denoting then $i$-th entry of vector $v$. 
From \eqref{eq:extstr}, $\bar{w}\in\cB_\infty$. Moreover,
$$
p'y-c=p'A^{-1}x-c-\gamma  p'A^{-1}G \bar{w} 
=p'A^{-1}x-c - \max\{p'A^{-1}x-c-1,0\} 
$$
which implies $0 \leq p'y-c \leq 1$ and hence $y\in\cS(p,c)$. The case in which $p'A^{-1}x-c \leq 0$ can be treated analogously.
\section{Proof of Theorem \ref{th:thresh}}
\label{app:th_thresh}

In order to prove the theorem, we need to introduce some lemmas.
\begin{lemma}
\label{lem1}
Let \eqref{eq:Cckn}-\eqref{eq:barlbarr} hold. Then, for any value taken by $y(k)$ one has
\begin{equation}
\rad \left( C\Xi(k|k-1)\bigcap C\cM(k) \right)=\frac{1}{2}\Delta(k) + \delta_v.
\label{eq:radint}
\end{equation}
Moreover, the interval $C\Xi(k|k-1)\bigcap C\cM(k)$ does not change if the feasible measurement set \eqref{eq:cMk} is rewritten as
\begin{equation}
\cM(k)=\left\{ x\in\rr^n:~|Cx-q(k)| \leq \frac{1}{2}\Delta(k) + \delta_v \right\}
\label{eq:cMkq}
\end{equation}
for some $q(k)\in\rr$.
\end{lemma}
\begin{proof}
From \eqref{eq:barlbarr} one has that $C\Xi(k|k-1)=[\bar{l}(k),~\bar{r}(k)]$.
By substituting \eqref{eq:thresh1} into \eqref{eq:cMk}, with $\tau_c$ and $\Delta$ given by \eqref{eq:Cckn}-\eqref{eq:Deltakn}, one gets
\begin{equation}
\begin{array}{rll}
\cM(k) =&
 \{ x\in\rr^n~: & Cx  \le \bar{l}(k)+\Delta(k) +2 \delta_v \hspace*{2cm}~\mbox{ if } y(k)=0; \vspace*{1mm}\\
&&  \bar{l}(k)+i \Delta(k)  < Cx  \le \bar{l}(k)+(i+1)\Delta(k) +2 \delta_v \vspace*{1mm}\\
&& \hspace*{3.5cm}~\mbox{ if } y(k)=i,~1\leq i \leq d-1;\vspace*{1mm}\\
&&  Cx> \bar{r}(k) -\Delta(k) - 2\delta_v   \hspace*{1.9cm}~\mbox{ if } y(k)=d \}.
\end{array}
\label{eq:cMk2}
\end{equation}
Then, \eqref{eq:radint} follows by noticing that the interval $C\Xi(k|k-1)\bigcap C\cM(k)$ has always the same length $\Delta(k) + 2\delta_v$, for every possible value taken by $y(k)$.
It is also easy to check that $C\Xi(k|k-1)\bigcap C\cM(k)$ does not change if in \eqref{eq:cMk2} we add the further constraints $Cx \geq \bar{l}(k)$ when $y(k)=0$ and $Cx \leq \bar{r}(k)$ for $y(k)=d$. Hence, $\cM(k)$ in \eqref{eq:cMk2} can be rewritten as in \eqref{eq:cMkq}.
\end{proof}

Consider the feasible measurement set $\cM(k-h)$, for $h=1,2,\dots$. By using \eqref{eq:sys1} iteratively, one can propagate such a set $h$ steps ahead, thus obtaining the set of states $x(k)$ compatible with the measurement $y(k-h)$, which we denote by $\cM(k|k-h)$ and is given by
\begin{equation}
\cM(k|k-h)= A^h\cM(k-h)+ \delta_w \sum_{i=0}^{h-1} A^{h-1-i} G \cB_\infty.
\label{eq:cMpred}
\end{equation}
\begin{lemma}
\label{lem2}
Let \eqref{eq:Cckn}-\eqref{eq:barlbarr} hold.  Then, the $h$-step ahead propagated measurement set in \eqref{eq:cMpred} takes on the form 
\begin{equation}
\label{eq:cMpred2}
\begin{array}{rl}
\cM(k|k-h)=\left\{ x\in\rr^n: \right. & \ds \left| CA^{-h}x-q(k-h) \right|   \\
& \left. \leq \ds \frac{1}{2}\Delta(k-h) + \delta_v + \delta_w \sum_{i=0}^{h-1} \| CA^{-(i+1)} G \|_1 \right\}.
\end{array}
\end{equation}
\end{lemma}
\begin{proof}
From Lemma \ref{lem1}, $\cM(k-h)$ can be written as in \eqref{eq:cMkq}, 
i.e. as the strip
$$
\cM(k-h)=\cS\left( \frac{C}{\frac{1}{2}\Delta(k-h) + \delta_v}, \frac{q(k-h)}{\frac{1}{2}\Delta(k-h) + \delta_v} \right).
$$
Then, \eqref{eq:cMpred2} follows by applying property \eqref{eq:prop} to \eqref{eq:cMpred}.
\end{proof}
\begin{lemma}
\label{lem3}
Let the sequence $x(k)\in\rr$ be such that
\begin{equation}
x(k) \leq a_1 x(k-1) +a_2 x(k-2) +\dots a_n x(k-n) + b
\label{eq:recx}
\end{equation}
with $a_i \geq 0$, $i=1,\dots, n$ and $b \geq 0$. Let the polynomial $z^n- \sum_{i=1}^n a_i z^{n-i}$ be Schur stable. Then, for any initial values $x(-h)$, $h=1,\dots,n$, one has
\begin{equation}
\limsup_{k \rightarrow +\infty}  x(k) \leq \frac{b}{1-\sum_{i=1}^n a_i}.
\end{equation}
\end{lemma}
\begin{proof}
Define the sequence $\bar{x}(k)$ such that
\begin{equation}
\bar{x}(k) = a_1 \bar{x}(k-1) +a_2 \bar{x}(k-2) +\dots a_n \bar{x}(k-n) + b
\label{eq:recbarx}
\end{equation} 
with $\bar{x}(-h) \geq x(-h)$, for $h=1,\dots,n$. Let $\psi(k)=\bar{x}(k)-x(k)$. From \eqref{eq:recx} and \eqref{eq:recbarx} one has
$$
\psi(k) \geq a_1 \psi(k-1) +a_2 \psi(k-2) +\dots a_n \psi(k-n).
$$
Being $\psi(-h) \geq 0$, for $h=1,\dots,n$, and $a_i \geq0$, $i=1,\dots,n$, one has $\psi(k) \geq 0$, $\forall k \geq 0$ and hence $x(k) \leq \bar{x}(k)$, $\forall k \geq 0$. 
From asymptotic stability of system \eqref{eq:recbarx}, one gets 
$$
\lim_{k \rightarrow +\infty}  \bar{x}(k) = \frac{b}{1-\sum_{i=1}^n a_i}
$$
from which the result follows.
\end{proof}

We are now ready to prove Theorem \ref{th:thresh}. The main idea is that, at each time $k$, one has
\begin{equation}
\Xi(k|k) \subseteq \bigcap_{h=0}^{n-1} \cM(k|k-h)
\label{eq:bound1}
\end{equation}
where $\cM(k|k)=\cM(k)$ given by \eqref{eq:cMkq}.
By noticing that both \eqref{eq:cMkq} and \eqref{eq:cMpred2} are strips, the right hand side in \eqref{eq:bound1} is a parallelotope. Hence, 
\begin{equation}
\Xi(k|k) \subseteq  \cP(k|k)= \cP(P(k),c(k)), 
\label{eq:bound2}
\end{equation}
where the shape matrix $P(k)$ is given by
\begin{equation}
P(k) = \Lambda(k)^{-1} \Omega
\label{eq:parmat}
\end{equation}
in which $\Lambda(k)=\mbox{diag}([\lambda(k|k) ~ \lambda(k|k-1) ~\dots~\lambda(k|k-n+1)]$, with
\begin{equation}
\lambda(k|k-h)= \ds \frac{1}{2} \Delta(k-h) + \delta_v + \delta_w \sum_{i=0}^{h-1} \| CA^{-(i+1)} G \|_1 ~,~~~ 
\label{eq:dkmh}
\end{equation}
for $h=0,1,\dots,n-1$. By nonsingularity of $A$ and observability of the system, one has that $\Omega$ is nonsingular and therefore the parallelotope $\cP(k|k)$ is bounded for every finite $k$. Moreover, by \eqref{eq:dkmh}, a sufficient condition for $\cP(k|k)$ to be asymptotically bounded is that the sequence $\Delta(k)$ defined by \eqref{eq:Deltakn}-\eqref{eq:barlbarr} is asymptotically upper bounded.
From \eqref{eq:barlbarr}, by using \eqref{eq:fss2} and \eqref{eq:bound2} one has  
\begin{eqnarray}
\bar{l}(k)& \geq &\inf_{x\in\cP(k-1|k-1),w\in\cB_\infty} CAx+\delta_wCGw
\\
\bar{r}(k)& \leq &\sup_{x\in\cP(k-1|k-1),w\in\cB_\infty} CAx+\delta_wCGw
\label{eq:barlbarr2}
\end{eqnarray}
Then, from \eqref{eq:Deltakn}, by applying the property \eqref{eq:projpar}, one gets
\begin{equation}
\Delta(k) \leq \frac{2}{d+1} \left\{ \|CA \,\Omega^{-1} \Lambda(k-1) \|_1 +\delta_w \|CG\|_1-\delta_v\right\}.
\label{eq:boundDelta}
\end{equation}
Now, letting $\alpha=CA\,\Omega^{-1}$ and exploiting \eqref{eq:dkmh}, one gets
\begin{equation}
\begin{array}{rcl}
\Delta(k) &\leq& \ds \frac{2}{d+1} \left\{ \sum_{h=0}^{n-1} |\alpha_{h+1}| \lambda(k-1|k-1-h)  +\delta_w \|CG\|_1-\delta_v\right\} \vspace*{2mm} \\
&=&  \ds \frac{2}{d+1} \left\{  \sum_{h=0}^{n-1} |\alpha_{h+1}|  \left[  \frac{1}{2} \Delta(k-1-h) + \delta_v + \delta_w \sum_{i=0}^{h-1} \| CA^{-(i+1)} G \|_1 \right] \right. \vspace*{1mm}\\
&& \hspace*{1.1cm} \left. +\delta_w \|CG\|_1-\delta_v\right\} \vspace*{2mm} \\
&=&  \ds \frac{1}{d+1}   \sum_{h=1}^{n} |\alpha_{h}| \Delta(k-h) + \bar{\gamma}
\end{array}
\label{eq:boundDelta2}
\end{equation}
with
$$
\bar{\gamma}= \ds \frac{2}{d+1}  \left\{ \delta_v\left(\|\alpha\|_1-1\right) + \delta_w \left( \|CG\|_1+ \sum_{h=1}^n |\alpha_h| \beta_h \right) \right\}
$$
and $\beta_h$ given by \eqref{eq:betah}.
If $d$ is chosen so that condition \eqref{eq:condp} holds, by using Lemma \ref{lem3} one obtains
$$
\limsup_{k \rightarrow +\infty} \Delta(k) \leq \Delta_\infty = \frac{\bar{\gamma}}{1- \frac{1}{d+1}\|\alpha\|_1}
$$
which corresponds to \eqref{eq:Deltainf}.
Then, from \eqref{eq:parmat}-\eqref{eq:dkmh} one gets $\ds\limsup_{k \rightarrow +\infty} \Lambda(k) \leq D_\infty$
given by \eqref{eq:Dinf}. Finally, \eqref{eq:rhoinf} follows from \eqref{eq:bound2} by using \eqref{eq:radpar}.

\bibliographystyle{ieeetran}
\bibliography{/Users/andreagarulli/newworld/Garulli/TEX/Bibtex/sm-FULL}

\end{document}